\documentclass[letterpaper, 10 pt, conference]{ieeeconf}  

\usepackage{amsthm,amsmath,amssymb,mathtools,newtxmath}
\usepackage{graphicx}
\usepackage{cite}
\usepackage{url}
\usepackage{cleveref}

\newtheorem{lemma}{Lemma}
\newtheorem{theorem}{Theorem}
\newtheorem{proposition}{Proposition}

\newtheorem{remark}{Remark}

\newtheorem{corollary}{Corollary}
\usepackage{algorithm,algpseudocode}
\usepackage{booktabs,lipsum}
\usepackage{subcaption}
\usepackage{color}

\DeclareMathOperator{\Tr}{Tr}
\DeclareMathOperator{\Id}{Id}

\IEEEoverridecommandlockouts                              

\title{\LARGE \bf
Assets Defending Differential Games with Partial Information and Selected Observations
}

\author{Yunhan Huang, Juntao Chen and Quanyan Zhu
\thanks{ Y. Huang and Q. Zhu are with the Department of Electrical and Computer Engineering,
        New York University, 370 Jay St., Brooklyn, NY 11201.
        {\tt\small \{yh.huang, qz494\}@nyu.edu}.}
\thanks{J. Chen is with the Department of Computer and Information Sciences,  Fordham University,  New York,  NY 10023.
        {\tt\small jchen504@fordham.edu}
        }
}

\begin{document}

\maketitle
\thispagestyle{plain}
\pagestyle{plain}

\begin{abstract}
In this paper, we consider a linear-quadratic-Gaussian defending assets differential game (DADG) where the attacker and the defender do not know each other's state information while they know the trajectory of a moving asset. Both players can choose to observe the other player's state information by paying a cost. The defender and the attacker have to craft both control strategies and observation strategies. We obtain a closed-form feedback solution that characterizes the Nash control strategies. We show that the trajectory of the asset does not affect both players' observation choices. Moreover, we show that the observation choices of the defender and the attacker can be decoupled and the Nash observation strategies can be found by solving two independent optimization problems. A set of necessary conditions is developed to characterize the optimal observation instances. Based on the necessary conditions, an effective algorithm is proposed to numerically compute the optimal observation instances. A case study is presented to demonstrate the effectiveness of the optimal observation instances.
\end{abstract}

\section{Introduction}
With recent advances in Autonomous Vehicles (AV) technologies, AV application scenarios emerge in modern military operations such as surveillance, persistent area denial \cite{liu2007pop}, pursuit-evasion \cite{weintraub2020introduction,talebi2017cooperative}, and assets defending\cite{li2011defending,garcia2020two}. Assets defending scenarios describe a setting where attackers attempt to intercept assets and defenders or interceptors strive for defending the assets. Assets defending scenarios pose challenging control design problems for AVs because AVs deployed often confront intelligent rivals with mobility and strategic decision making. Differential game theory offers the right set of theoretical underpinnings to investigate assets defending scenarios and to develop optimal strategies for each player. Hence, several works have addressed different assets defending problems formulated as differential games\cite{li2011defending,garcia2018design,garcia2020defense,garcia2020two,liang2019differential,pachter2017differential}. The formulations often times are referred to as Defending Assets Differential Games (DADG). Among the various DADG models adopted, the linear quadratic differential game formulation is favored due to its analytical friendliness \cite{li2011defending,garcia2020defense}. A common assumption taken for granted in previous studies of DADG is that state information is freely available any time to both the attacker and the defender. However, in reality, state information, especially information regarding one's opponent, is not accessible and usually is expensive to acquire. For example, in naval warfare, the detection of aircraft carriers is a challenging task considering the vastness of the ocean and the search for adversarial submarines due to its stealthiness. The detection of such military units comes with monetary expenses, risks of exposing oneself, and loss of surveillance aircraft.

To fill the vacuum, in this work, we study DADG with partial information where the attacker and the defender only have access to their own state information. The attacker does not know the state information of the defender, unless the attacker chooses to observe the defender's information by paying a cost. So does the defender. We assume that both the defender and the attacker know the trajectory of the asset. The DADG is a dynamic game with a pre-specified duration. Since the two players are constantly moving, the information obtained earlier may deteriorate over time and a new observation need to be made. Thus, both players have to decide when to observe and how many times to observe within the duration. At the same time, control strategies need to be developed based on the observed information. 

The join design of observation strategies and control strategies with costly observations have been investigated by several papers \cite{cooper1971optimal,huang2019continuous,olsder1977observation,huang2020infinite,maity2016strategies,maity2017linear,huang2020cross}. In 70s, Cooper studied an discrete-time optimal control problem where the controller decides at each step whether to observe an noisy observation or not \cite{cooper1971optimal}. Olsder later extended this study into a two-player discrete-time dynamic game setting where each player chooses when to observe and a solution is obtained for a two-stage dynamic game\cite{olsder1977observation}. More recently, \cite{huang2019continuous} looked into controlled observation for continuous-time Markov decision processes, where applications in queueing systems and inventory management have been studied. In \cite{maity2016strategies} and \cite{maity2017linear}, Maity et al. extends the problem formulated in \cite{olsder1977observation} to a linear quadratic continuous-time setting, where each player has no state information at all unless they choose to observe. But in \cite{maity2016strategies} and \cite{maity2017linear}, the solution of the observation strategies are not provided and only some properties regarding the solution are identified. \cite{huang2020cross} studies a discrete-time dynamic game with controlled yet noisy observations where one player acts as a jammer that intercepts the observation of the other player.

Our work differs from the previous works in three ways: First, we focus on DADG which has different cost structure and system dynamics from previous works. Second, both the attacker and the defender knows their own information, which gives each player partial information even when they choose not to observe. They obtain full information when they choose to observe. Both players do not share their information, which causes asymmetry of information. Third, we characterize the control strategies fully and develop an effective algorithm that numerically calculate the observation instances. 

The contributions of this paper is summarized as follows. First, we abandon the common yet unrealistic assumption that information is freely available all the time in previous DADG works. A linear-quadratic-Guassian DADG framework with controlled partial information is proposed. Second, we fully characterize the Nash control strategies and develop a set of necessary conditions that characterize the optimal observation strategy. We shows the separation principle, in which the observation choices only affect the state estimate in the control strategies. Analytical results show that the observation choices are independent from the trajectory of the asset. We further show that the observation decision of the attacker and the observation decision of the defender are decoupled. Hence, either the attacker or the defender can make observation choices without anticipating each other's choices. As a result, the Nash observation strategies can be obtained by solving two independent optimization problems. Third, we develop an effective algorithms that can numerically compute the optimal observation instances. We demonstrate the effectiveness of the optimal observation instances by comparing it with the periodic observation instances.

The rest of the paper is organized as follows. In Sec. \ref{sec:ProblemForm}, we formulate the linear-quadratic-Gaussian DADG with controlled information. Sec. \ref{Sec:TheoreticalResults} gives the theoretical results regarding the Nash control strategies and the Nash observation strategies. In Sec. \ref{sec:CaseStudy}, we conduct a case study to demonstrate the effectiveness of the optimal observation instances.

\subsection{Notation}
For a vector $x$ or a matrix $M$, $x'$ and $M'$ represents the transpose of the vector $x$ and the matrix $M$ respectively. An $n\times n$ identity matrix is denoted by $\Id_n$. $\Vert \cdot \Vert_2$ is the $L$-$2$ norm. For a vector with proper dimension, the norm $\Vert \cdot \Vert_Q$ is defined as $\Vert x\Vert_Q = x'Qx$. The set of all real numbers is denoted by $\mathbb{R}$ and $\mathbb{N}$ denotes the set of all natural numbers. The trace operator is denoted by $\Tr(\cdot)$. The Kronecker product is represented by $\otimes$.

\section{Defending an Asset with Controlled Information}\label{sec:ProblemForm}

In this section, we formulate a DADG with controlled observation. The dynamics of each player is described by the following linear systems 
\begin{equation}\label{eq:SeparateDynamics}
\begin{aligned}
d{x}_a &= A_a x_a(t)dt + \tilde{B}_a u_a(t)dt + C_a {dw_a(t)},\ \ \ x_a(0) = x_{a0}\\
d{x}_d &= A_d x_d(t)dt + \tilde{B}_d u_d(t)dt + C_d {dw_d(t)},\ \ \ x_d(0) = x_{d0}\\
\end{aligned}
\end{equation}
where $x_a\in\mathbb{R}^n,x_d\in\mathbb{R}^n$ are states of the attacker and the defender; $u_a\in U_a,u_d \in U_d$ are controls inputs the corresponding players; $w_a$ and $w_d$ are independent standard Wiener processes. $A_a,A_d,\tilde{B}_a,\tilde{B}_d,C_a$, and $C_d$ are real matrices with proper dimensions. The time index is denoted by $t$ and the consider a finite-time horizon $[0,t_f]$. Let $x_s(t)$ be the location of the asset at time $t$ and the trajectory of the asset is given and known to both the attacker and the defender. In this paper, we consider the cases of a stationary asset and an asset with an arbitrary trajectory. We assume that there is an auxiliary linear system that captures the trajectory $x_s(\cdot)$ of the asset
$
\dot{x}_s = A_s x_s,\ x_s(0) = x_{s0}.
$ 
This assumption is introduction for analysis purpose and is not necessary, as we will show later. From a systematic point of view, we can formulate an aggregate system as
\begin{equation}\label{eq:AggregatedStateDynamics}
d{x}(t) = Ax(t)dt  +  B_a u_adt + B_d u_ddt  + C dw(t),\ \textrm{with }x(0) = x_0,
\end{equation}
where $x = [x_a'\ x_d'\ x_s']'$, $B_a = [\tilde{B}_a'\ \mathbf{0}\ \mathbf{0}]'$, $B_d = [\mathbf{0}\ \tilde{B}_d'\ \mathbf{0}]'$, $w = [w_a'\ w_d'\ \mathbf{0}]'$,
$$
A= \begin{bmatrix}
A_a& \mathbf{0} & \mathbf{0} \\
\mathbf{0} & A_d & \mathbf{0} \\
\mathbf{0} & \mathbf{0} & A_s
\end{bmatrix},\ \ \ \text{and } C = \begin{bmatrix}
C_a& \mathbf{0} & \mathbf{0} \\
\mathbf{0} & C_d & \mathbf{0}\\
\mathbf{0} & \mathbf{0} & \mathbf{0}
\end{bmatrix}.
$$

In this paper, we consider a situation where the defender and the attacker can select a set of time instances to observe one's opponent state. The information structure of the attacker and the defender is summarized as follows: 1. Both the defender and the attacker know the trajectory of the asset. 2. The defender and the attacker know their own state, but they don't know each other's state. 3. Each player can choose to observe the other player's state by paying a cost.

Let $\mathcal{T}_a = \{t_{1,a},t_{2,a},\cdots,t_{N_a,a}\}$ be the set of time instances when the attacker choose to observe. Let $\mathcal{T}_d =\{ t_{1,d},t_{2,d},\cdots, t_{N_d,d}\}$ be that of the defender. Here, $N_a$ and $N_d$ are the number of observations made by the attacker and the defender respectively within time horizon $[0,t_f]$. Let $y_a(t)$ and $y_d (t)$ be the observations of the attacker and the defender respectively. We can write the above description of the information structure as
$$
\begin{aligned}
y_a(t)  &= \begin{bmatrix}
\Id_n & \mathbf{0} & \mathbf{0}\\
\mathbf{0} & \mathbf{0} & \mathbf{0}\\
\mathbf{0} & \mathbf{0} & \Id_n
\end{bmatrix}x(t),\ \ \ \textrm{for }t\notin \mathcal{T}_a\\
y_d(t) &= \begin{bmatrix}
\mathbf{0} & \mathbf{0} & \mathbf{0}\\
\mathbf{0} & \Id_n & \mathbf{0}\\
\mathbf{0} & \mathbf{0} & \Id_n
\end{bmatrix}x(t),\ \ \ \textrm{for }t\notin \mathcal{T}_d,\\
y_a(t) &= x(t)\ \ \ \textrm{for }t\in\mathcal{T}_a,\ \ \  y_d(t) = x(t),\ \ \ \textrm{for }t\in  \mathcal{T}_d.
\end{aligned}
$$
Let $\mathcal{I}_a(t)$ be the information the attacker has at time $t$ and $\mathcal{I}_d(t)$ be that of the defender. Given $\mathcal{T}_a,\mathcal{T}_d$, We have
\begin{equation}\label{eq:InformationSet}
\mathcal{I}_a(t) = \{x_0, y_a(\tau), \tau\leq t\},\ 
\mathcal{I}_d(t) = \{x_0, y_d(\tau), \tau\leq t\}.
\end{equation}
Each player considers stationary feedback strategies $\gamma_a$ and $\gamma_d$ such that $u_a(t) = \gamma_a(\mathcal{I}_a(t))$ and $u_d(t) = \gamma_d(\mathcal{I}_d(t))$. We consider the objective function of the following form
\begin{equation}\label{Eq:CostFunctional}
\begin{aligned}
J(&\gamma_a,\gamma_d;x_0)\\ 
 =& \mathbb{E}\bigg[ \int_0^{t_f} \Big(u_a(t)'u_a(t) - u_d(t)'u_d(t) + \omega_a^I \Vert x_a(t) - x_s(t) \Vert^2_2\\
 &- \omega_d^I \Vert x_d(t) - x_a(t) \Vert_2^2  \Big)dt + \omega_a\Vert x_a(t_f) - x_s(t_f) \Vert_2^2\\
 &- \omega_d \Vert x_d(t_f)  - x_a(t_f) \Vert_2^2 + O N_a - O N_d\bigg \vert\ x(0) =x_0\bigg],\\
\end{aligned}
\end{equation}
where $\omega^I_a$, $\omega^I_d$, $\omega_a$, and $\omega_d$ are the weighting coefficients that captures the trade-off between the attacker-to-defender distance and the attacker-to-asset distance. The superscript $I$ indicates that the weight is for the intermediate cost rather than the terminal cost. In some DADG papers, $w_a^I$ and $w_d^I$ are set to be zero and hence attention is paid to the terminal state alone\cite{garcia2020two}. The scalar $O\geq 0$ is the cost of making observations. The attacker tries to minimize its distance to the asset while trying to avoid being intercepted by the defender. Hence, the attacker aims to minimize the objective function and the defender, however, aims to maximize it. Even though we consider the state of all players lying in the same space $\mathbb{R}^n$, the results in this paper can be easily extended to a general setting.  It is tacitly assumed that the system characteristics are know to both players.

This formulation gives us a differential game with asymmetric yet controlled information. The problem formulation brings up a series of questions: When does the defender need to observe the attacker's state? For the defender, is it worth paying a cost to observe the attacker's state while the defender knows that the attacker is tracking the asset and the location of the asset is known to the defender? Are the optimal observation instances dependent on the trajectory of the asset?  In the next section, we develop our main results that address these questions.

\section{Theoretical Results}\label{Sec:TheoreticalResults}

In this section, we develop our main theoretical results of this paper. A close look at (\ref{Eq:CostFunctional}) gives the following form
\begin{equation}\label{Eq:CostFunctionalCompact}
\begin{aligned}
J(&\gamma_a,\gamma_d;x_0)\\ 
 =& \mathbb{E}\bigg[ \int_0^{t_f} \Big(u_a(t)'u_a(t) - u_d(t)'u_d(t) + \Vert x(t)\Vert_{Q}^2  \Big)dt\\
 &+ \Vert x(t_f)\Vert_{Q_{f}}^2 + O N_a - O N_d\bigg \vert\ x(0) =x_0\bigg],\\
\end{aligned}
\end{equation}
where $Q = \tilde{Q}(\omega^I_a,\omega^I_d)$, $Q_{f} = \tilde{Q}(\omega_a,\omega_d)$ with
$$
\tilde{Q}(\omega^I_a,\omega^I_d) = \begin{bmatrix}
(\omega_a^I - \omega_d^I) \Id_n & \omega^I_d \Id_n  & -\omega^I_a \Id_n\\
\omega_d^I \Id_n & -\omega^I_d \Id_n &  \mathbf{0}\\
-\omega^I_a \Id_n & \mathbf{0} & \omega_a^I\Id_n
\end{bmatrix}.
$$
\subsection{The Nash Control Strategies}
Applying It\^{o}'s lemma and a completion of squares on (\ref{Eq:CostFunctionalCompact}) yields the following lemma.
\begin{lemma}
    The cost functional $J$ in (\ref{Eq:CostFunctionalCompact}) with state dynamics (\ref{eq:AggregatedStateDynamics}) has the following form
    \begin{equation}\label{eq:CostFunctionalSquared}
    \begin{aligned}
        J = &\mathbb{E}\bigg[ \int_{0}^{t_f} \Vert u_a(t) + {B}_a' K(t)x(t) \Vert^2_2 - \Vert {u_d - {B}_d' K(t)x(t)} \Vert_2^2dt\\ 
            & + \Vert x_0\Vert^2_{K(0)} + \int_0^{t_f}\Tr\left( K(t) CC'\right)dt + ON_a - O N_d \bigg],\\
    \end{aligned}
    \end{equation}
    where $(K(t),t\in[0,t_f])$ is symmetric and satisfies the Riccatic equation
    \begin{equation}\label{eq:RiccatiEquation}
        \dot{K}(t) =- K(t)A - A'K(t) -Q - K(t)\left(B_dB_d' - B_a B_a' \right)K(t)
    \end{equation}
    with $K(t_f) = Q_{f}$.
\end{lemma}
The proof follows standard arguments of the ``completion of squares'' procedures. Readers are referred to \cite[Theorem 1]{chen2019covariance}, \cite[Theorem 3.1]{maity2016strategies}, or \cite[Lemma 1]{huang2021pursuit} for specifics. The existence of bounded solutions for (\ref{eq:RiccatiEquation}) depends on $Q$, $Q_{f}$, and $B_dB_d' - B_a B_a'$, which we shall discuss later. Note that the attacker aims to find an observation-dependent strategy $\gamma_a$ to minimize (\ref{eq:CostFunctionalSquared}) while the defender desires to maximizes (\ref{eq:CostFunctionalSquared}). From $(\ref{eq:CostFunctionalSquared})$, we know that the choice of strategies only affect the terms within the first integral. Therefore, given $\mathcal{T}_a$ and $ \mathcal{T}_d$, the Nash control strategies will be of the form
\begin{equation}\label{eq:NashControlForm}
u_a^*(t) =  -B_a'K(t)\hat{x}_1(t),\ \ \ 
u_d^*(t) = B_d' K(t)\hat{x}_2(t),
\end{equation}
for some $\hat{x}_1$ and $\hat{x}_2$. The choices $\hat{x}_1$ and $\hat{x}_2$ are made by the attacker and the defender respectively such that $u^*_a(t)$ is $\mathcal{I}_a(t)$ measurable and $u^*_d(t)$ is $\mathcal{I}_d(t)$ measurable. To decompose the Nash control strategies, we split the $3n\times 3n$ matrix in (\ref{eq:NashControlForm}) as 
$$
K(t) = \begin{bmatrix}
K_{11}(t) & K_{12}(t)\\
K_{12}(t)' & K_{22}(t)\\
\end{bmatrix},
$$
where $K_{11}$ is an $2n\times 2n$ matrix function of time $t$. Similarly, $Q$ and $Q_{f}$ are partitioned into $Q_{ij}$ and $Q_{{f},ij}$ for $i,j\in\{1,2\}$. To decompose the Nash control strategy, we derive the following decomposed Riccati equations
\begin{align}
\begin{split}
\dot{K}_{11} = -K_{11}\hat{A} - &\hat{A}'K_{11}  - Q_{11} - K_{11}\left(\hat{B}_d \hat{B}_d' - \hat{B}_a \hat{B}_a'\right)K_{11},\\
&\textrm{with }K_{11}(t_f) = Q_{f,11},\\
\end{split} \label{eq:SubRiccatiEquation} \\
\begin{split}
\dot{K}_{12} = -K_{12} A_s - &\hat{A}'K_{12} - Q_{12} - K_{12}\left(\hat{B}_d \hat{B}_d' - \hat{B}_a \hat{B}_a'\right)K_{11},\\
&\textrm{with }K_{12}(t_f) = Q_{f,12},\\
\end{split}\nonumber
\end{align}
where 
\begin{equation}\label{eq:TruncatedSystemPara}
\hat{A} = \begin{bmatrix}
A_a & \mathbf{0}\\
\mathbf{0} & A_d\\
\end{bmatrix},\ \ \hat{B}_a = \begin{bmatrix}
{\tilde{B}_a}\\ \mathbf{0}
\end{bmatrix},\ \ {\hat{B}_d = \begin{bmatrix}
\mathbf{0} \\ \tilde{B}_d
\end{bmatrix}.}
\end{equation}
\begin{theorem}\label{theorem:NashControlStrategies}
Suppose that $\mathcal{T}_a$ and $\mathcal{T}_d$ are known and the trajectory $(x_s(\tau),\tau\in[0,t_f])$ is given.  The DA game defined by (\ref{eq:AggregatedStateDynamics}) and (\ref{eq:CostFunctionalSquared}) admits a Nash control strategy 
\begin{align}
\begin{split}
{u}^*_a &= -\hat{B}_a' K_{11} \begin{bmatrix}
{x_a}\\ \hat{x}_{1,d} 
\end{bmatrix}
-\hat{B}_a' s,\\
\end{split} \label{eq:AttackerControl} \\
\begin{split}
u_d^* &= \hat{B}_d K_{11}\begin{bmatrix}
\hat{x}_{2,a}\\ x_d
\end{bmatrix} + \hat{B}_d' s,\\
\end{split}\label{eq:DefenderControl}
\end{align}
where $(K_{11},t\in [0,t_f])$ is the solution of the Riccati equation (\ref{eq:SubRiccatiEquation}), $\hat{A}$, $\hat{B}_a$, and $\hat{B}_d$ are defined in (\ref{eq:TruncatedSystemPara}), and $(s(t),t\in[0,t_f])$ is generated by
\begin{equation}
    \dot{s}=\left[-\hat{A}' - K_{11}\left(\hat{B}_d \hat{B}_d' - \hat{B}_a \hat{B}_a' \right) \right]s - Q_{12}x_s,
\end{equation}
with $s(t_f) = Q_{f,12}x_s(t_f)$. Moreover, the estimate $\hat{x}_{1,d}$ of defender's state evolves as
\begin{equation}\label{eq:EstimateDynamicsAttacker}
\dot{\hat{x}}_{1,d}= A_d \hat{x}_{1,d} + \tilde{B}_d\hat{B}_d'\left(K_{11}\begin{bmatrix} x_a\\ \hat{x}_{1,d} \end{bmatrix} + s \right)
\end{equation}
with $\hat{x}_{1,d}(0)= x_{d0}$ and $\hat{x}_{1,d}(t) = x_d(t)$ for every $t\in\mathcal{T}_a$, and 
\begin{equation}\label{eq:EstimateDynamicsDefender}
   \dot{\hat{x}}_{2,a} = A_a \dot{\hat{x}}_{2,a} - \tilde{B}_a \hat{B}_a'\left( K_{11} \begin{bmatrix}
\hat{x}_{2,a}\\
x_d
\end{bmatrix} + s\right),
\end{equation}
with $\hat{x}_{2,a}(0)= x_{a0}$ and $\hat{x}_{2,a}(t) = x_a(t)$ for all $t\in \mathcal{T}_d$.
\end{theorem}
\begin{proof}

In the proof, we drop the time index $t$ in some places for convenience. From \cite[Proof of Theorem 2]{maity2017linear}, we know that consider the game defined by (\ref{Eq:CostFunctionalCompact}) and (\ref{eq:AggregatedStateDynamics}), (\ref{eq:NashControlForm}) constitutes a Nash control strategy if
{\begin{equation}\label{eq:NecessaryConEstimate}
\begin{aligned}
&\mathbb{E}\left[x(t) - \hat{x}_1(t)\middle\vert \mathcal{I}_a(t) \right] = 0,\ \mathbb{E}\left[ x(t) -\hat{x}_2(t) \middle\vert \mathcal{I}_a(t) \right] = 0,\ \textrm{and}\\
&\mathbb{E}\left[ x(t) - \hat{x}_2(t)  \middle\vert \mathcal{I}_d(t)  \right] = 0,\ \mathbb{E}\left[x(t) - \hat{x}_1(t)| \mathcal{I}_d(t)\right] =0.
\end{aligned}
\end{equation}}
With this result, if we let $\hat{x}_1 = [\hat{x}_{1,a}'\ \hat{x}_{1,d}'\ \hat{x}_{1,s}']'$ and $\hat{x}_2 = [\hat{x}_{2,a}'\ \hat{x}_{2,d}'\ \hat{x}_{2,s}']'$, we have $\hat{x}_{1,a} = x_a$, $\hat{x}_{1,d} = \mathbb{E}\left[x_d|\mathcal{I}_a\right]$, $\hat{x}_{1,s}=x_s$ and $\hat{x}_{2,a} =\mathbb{E}[ x_a|\mathcal{I}_d]$  $\hat{x}_{2,d} = x_d$, $\hat{x}_{2,s} = x_s$.

Using (\ref{eq:SubRiccatiEquation}), the Nash control strategy of the attacker can be decomposed as
$$
{u}^*_a = -\hat{B}_a' K_{11} \begin{bmatrix}
{x_a}\\ \hat{x}_{1,d} 
\end{bmatrix}
-\hat{B}_a' K_{12} x_s,
$$
where the attacker's control is driven by his/her state, his/her estimate of the defender's state, as well as the trajectory of the asset. To eliminate the dependence of the control on $A_s$ (which is introduced for auxiliary purpose), we let $s(t) = K_{12}(t)x_s(t)$.  Note that 
$$
\begin{aligned}
\dot{s} &= \dot{K}_{12}x_s + K_{12} \dot{x}_s\\
&= \left[-\hat{A}' - K_{11}\left(\hat{B}_d \hat{B}_d' - \hat{B}_a \hat{B}_a' \right) \right]s - Q_{12}x_s,
\end{aligned}
$$
with $s(t_f) = Q_{f,12} x_{s}(t_f)$. Hence, the control  depends on the trajectory of the asset irrespective of $A_s$. 
Under the Nash control strategy, the dynamics of the attacker is
$$
\dot{x}_a = A_a x_a(t) - \tilde{B}_a\hat{B}_a' K_{11} \begin{bmatrix}
x_a'\\ \hat{x}_{1,d} 
\end{bmatrix} - \tilde{B}_a\hat{B}_a's + C_a dw_a(t),
$$
with $x_a(0) = x_{a0}$.  From (\ref{eq:NecessaryConEstimate}), and using the fact $\mathbb{E}[w_a(t)|\mathcal{I}_a(t)]=0$, we have
$$
\begin{aligned}
\dot{\hat{x}}_{2,a} &= A_a \dot{\hat{x}}_{2,a} - \tilde{B}_a \hat{B}_a'\left(K_{11} \begin{bmatrix}
\hat{x}_{2,a}\\
\mathbb{E}[\hat{x}_{1,d}|\mathcal{I}_d] 
\end{bmatrix} + s\right)\\
&= A_a \dot{\hat{x}}_{2,a} - \tilde{B}_a \hat{B}_a'\left( K_{11} \begin{bmatrix}
\hat{x}_{2,a}\\
x_d
\end{bmatrix} + s\right)
\end{aligned}
$$
with $\hat{x}_{2,a}(0)= x_{a0}$ and $\hat{x}_{2,a}(t) = x_a(t)$ for every $t\in \mathcal{T}_d$.
Hence, we fully characterizes the Nash control strategy of the defender (\ref{eq:DefenderControl}).
Similarly, we can obtain the attacker's estimate of the defender's state, which is given by (\ref{eq:EstimateDynamicsDefender}).
\end{proof}
\begin{remark}
The attacker's estimate of the defender's state $\hat{x}_{1,d}$ evolves according to (\ref{eq:EstimateDynamicsAttacker}). The estimate $\hat{x}_{1,d}$ does not require the attacker to know the control of the defender. Every time the attacker choose to observe, he/she receives the actual state of the defender, i.e., $\hat{x}_{1,d}(t)=x_d(t),\forall t\in\mathcal{T}_a$. The solution of the Riccati equation (\ref{eq:SubRiccatiEquation}) may admits a finite escape time since the conditions that $Q_{11}$ is positive semi-definite and $\hat{B}_d \hat{B}_d' - \hat{B}_a \hat{B}_a'$ is positive-definite do not hold. We can use a more lenient condition given by \cite[Corollary 5.13]{engwerda2005lq} to check the existence of a bounded solution of (\ref{eq:SubRiccatiEquation}). Due to the space constraints, we do not restate the corollary here. Instead, we provide a closed-form bounded solution of the Riccati equation for our case study in Sec. \ref{sec:CaseStudy}.
\end{remark}

\subsection{The Nash Observation Choices}
In Theorem \ref{theorem:NashControlStrategies}, we characterize the Nash control strategies of both players when the observation instances are given. To understand how both players would select their observation instances, we need to obtain the cost functional under the Nash control strategies for any given $\mathcal{T}_a$ and $\mathcal{T}_b$. From the decomposition of the Riccati equation in (\ref{eq:SubRiccatiEquation}), we know
$$
B_{a}' K x = \hat{B}_a' K_{11}\begin{bmatrix}
x_a \\ x_d
\end{bmatrix} + {\hat{B}_a' s.}
$$
We further decompose $K_{11}$ into
$$
K_{11} = \begin{bmatrix}
K_{11}^{ul} & K_{11}^{ur}\\
{K_{11}^{ur}}' & K_{11}^{br}
\end{bmatrix},
$$
where $K_{11}^{ul}$, $K_{11}^{ur}$, and $K_{11}^{br}$ are $n\times n $ matrices.
Then, the first term in (\ref{eq:CostFunctionalSquared}) can be written as
$$
\begin{aligned}
\left\Vert u^*_a(t) + {B}_a' K(t)x(t) \right\Vert^2_2 &= \left\Vert \hat{B}_a' K_{11} \begin{bmatrix} 0\\ x_d -\hat{x}_{1,d}\end{bmatrix} \right\Vert_2^2\\
&=\Vert \tilde{B}_a' K_{11}^{ur}(x_d - \hat{x}_{1,d})\Vert_2^2.
\end{aligned}
$$
Similarly, we obtain 
$$
{\left\Vert u^*_d(t) - B_d' K(t) x(t)  \right\Vert_2^2} = \Vert \tilde{B}_d' {K_{11}^{ur}}'(x_a - \hat{x}_{2,a})  \Vert_2^2. 
$$
From (\ref{eq:SeparateDynamics}) and (\ref{eq:EstimateDynamicsAttacker}), we know that 
\begin{equation}\label{eq:EstimationError}
\begin{aligned}
&\mathbb{E}[(x_d(t) -\hat{x}_{1,d}(t))(x_d(t) - \hat{x}_{1,d}(t))'] \\
=& {\int_{\tilde{t}}^t e^{A_d(s-\tilde{t})} C_d C_d'e^{A_d(s-\tilde{t})'} ds,}
\end{aligned}
\end{equation}
where $\tilde{t}$ the latest observation before $t$, which is dependent on $t$ and $\mathcal{T}_a $ and is defined as $\tilde{t} = \max \{\tau\ |\ \tau\in\mathcal{T}_a, \tau \leq t \}$. The discussion above leads to the following corollary.
\begin{corollary}\label{corol:CostUnderNashControl}
For given $\mathcal{T}_a$ and $\mathcal{T}_b$, let $\mathcal{T}_a =\{t_{1,a}.t_{2,a},\cdots,t_{N_a,a}\}$ and $\mathcal{T}_d = \{t_{1,d},t_{2,d},\cdots,t_{N_d,d}\}$ with $t_{1,a}<t_{2,a}<\cdots < t_{N_a,a}$ and $t_{1,d}<  t_{2,d} <\cdots < t_{N_d,d}$. Under the Nash control strategies obtained in Theorem \ref{theorem:NashControlStrategies}, the cost functional (\ref{eq:CostFunctionalSquared}) becomes
\begin{equation}\label{eq:CostUnderNashControl}
\begin{aligned}
&J(\gamma^*_a, \gamma^*_d,x_0)\\ 
=& \sum_{i=0}^{N_a} \int_{t_{i,a}}^{t_{i+1,a}} \Tr\left[ \Sigma_{1,d}(t-t_{i,a}) \varphi_a(t) \right]dt\\
&\ -\sum_{i=0}^{N_d} \int_{t_{i,d}}^{t_{i+1,d}} \Tr\left[ \Sigma_{2,a}(t-t_{i,d}) \varphi_d(t)\right] dt\\
&\ + \int_0^{t_f}\Tr\left( K(t) CC'\right)dt  + \Vert x_0\Vert^2_{K(0)} +  ON_a - O N_d ,\\
\end{aligned}
\end{equation}
where $\sum_{1,d}(t) = \int_{0}^t e^{A_d \tau}C_d C_d'e^{A_d \tau}d\tau$, $\Sigma_{2,a}= \int_{0}^t e^{A_a \tau}C_a C_a'e^{A_a \tau}d\tau$, $t_{i,a} \in \mathcal{T}_a$, $t_{0,a} = t_{0,d} = 0$, and $t_{N_a+1,a} = t_{N_d+1,d} =t_f$. Moreover, $\varphi_a(t) = {K_{11}^{ur}}(t)' \tilde{B}_a \tilde{B}_a' K_{11}^{ur}(t)$ and $\varphi_d(t) =  {K_{11}^{ur}}(t) \tilde{B}_d \tilde{B}_d{'} {K_{11}^{ur}}(t)'$.
\end{corollary}
Corollary \ref{corol:CostUnderNashControl} presents the cost functional under the Nash control strategies. Among the six terms in (\ref{eq:CostUnderNashControl}), only the first and the last two terms are associated with $\mathcal{T}_a$ and $\mathcal{T}_d$. Note that the objective of the attacker is to find a set of observation instances $\mathcal{T}_a$ that minimizes $J$, while the defender aims to maximize $J$.  Hence, to decide their observation instances, the attacker and the defender only have to consider the first two terms and the last two terms. Moreover, the effect of the two players' observation can be decoupled, by which we mean 
$$
\tilde{J}(\mathcal{T}_a,\mathcal{T}_b) = \tilde{J}_a(\mathcal{T}_a) - \tilde{J_b}(\mathcal{T}_b),
$$
where
{\small
\begin{equation}\label{eq:SeparatedObservationCost}
\begin{aligned}
 \tilde{J}_a(\mathcal{T}_a)&=  \sum_{i=0}^{N_a} \int_{t_{i,a}}^{t_{i+1,a}} \Tr\left[ \Sigma_{1,d}(t-t_{i,a}) \varphi_a(t)\right]dt + O N_a,\\
 \tilde{J}_d(\mathcal{T}_d)&=\sum_{i=0}^{N_d} \int_{t_{i,d}}^{t_{i+1,d}} \Tr\left[ \Sigma_{2,a}(t-t_{i,d}) \varphi_d(t)\right] dt + O N_d.
\end{aligned}
\end{equation}
}
\begin{remark}
Corollary \ref{corol:CostUnderNashControl} shows that the optimal choices of observation instances do not depend on the trajectory of the asset, by which we mean no matter how the asset moves, the defender and the attacker's choices of observation instances will not be affected. This is due to the fact that both players know the trajectory of the asset. The relative position between the asset and the attacker can be estimated unbiasedly by the defender. The cost is captured by the variance of the estimate error which is independent of the asset's trajectory. 
\end{remark}
\begin{remark}
Since $\tilde{J}(\mathcal{T}_a,\mathcal{T}_d)$ can be decomposed into $\tilde{J}_a(\mathcal{T}_a) - \tilde{J}_d(\mathcal{T}_d)$, the Nash observation strategies that solve $\min_{\mathcal{T}_a}\max_{\mathcal{T}_d} \tilde{J}(\mathcal{T}_a,\mathcal{T}_d)$ can be obtained by solving $\min_{\mathcal{T}_a} \tilde{J}_a(\mathcal{T}_a)$ and $\min_{\mathcal{T}_b} \tilde{J}_{d}(\mathcal{T}_d)$. That means the defender and the attacker can make independent observation choices by solving two independent optimization problems. The independence comes from the fact that the defender and the attacker have independent dynamics in (\ref{eq:SeparateDynamics}).
\end{remark}
\begin{remark}
To provide insights on the cost functional, we take the attacker as an example. In the first term of $\tilde{J}_a(\mathcal{T}_a)$, we know
\begin{equation}\label{eq:TraceInsight}
\begin{aligned}
&\Tr\left[ \Sigma_{1,d}(t-\tilde{t}) \varphi_a(t)\right]\\
=& \mathbb{E}[(x_d(t)-\hat{x}_{1,d}(t))' {K_{11}^{ur}}' \tilde{B}_a \tilde{B}_a' K_{11}^{ur}(x(t)-\hat{x}_{1,d}(t))],
\end{aligned}
\end{equation}
where $\Sigma_{1,d}(t-\tilde{t})$ is the variance of the estimation error $x_d - \hat{x}_{1,d}$ at time and $\tilde{t}$ is the latest observation instance before time $t$. The term in (\ref{eq:TraceInsight}) captures the the instantaneous cost at time $t$ induced by the mismatch between the actual state of the defender and the attacker's estimate. The observation choices are control-aware by which we mean the estimation error is scaled by the matrix ${K_{11}^{ur}}' \tilde{B}_a \tilde{B}_a' K_{11}^{ur}$ and he matrix assign more weight to the estimation error corresponding to the states that are more informative to control needs.
From (\ref{eq:EstimationError}), we know that the estimation error accumulates until the attacker makes an observation. As a result, the observation clears the estimation error. However, each observation made is subject to a cost $O$. Hence, the attacker has to make observation decision strategically over time. Overall, the observation decision has to consider the trade-off between the estimation error and the number of observations. Moreover, the observation instances need to be well designed by both players to minimize the corresponding integral terms in (\ref{eq:SeparatedObservationCost}).
\end{remark}

Since solving the Nash observation game is equivalent to solving two independent optimization problems, we focus on solving the attacker's optimization problem. One can obtain the result for the defender similarly. The solution of $\min_{\mathcal{T}_a}\tilde{J}_a(\mathcal{T}_a)$ involves two components: the optimal number of observations $N_a^*$ and a set of optimal observation instances $\mathcal{T}^*_a=\{t^*_{1,a}, t^*_{2,a},\cdots, t^*_{N^*_a,a}\}$ (with a slight abuse of notation here). Define
\begin{equation}\label{eq:ObsrvationInstancesOp}
\begin{aligned}
f^*_a(N_a) \coloneqq &\min_{t_1.\cdots,t_{N_a}} f_a(t_1,t_2,\cdots,t_{N_a})\\
&\ \ \ s.t.\ \ \ t_0 = 0,\ t_{N_a} = t_f,\\
&\ \ \ \ \ \ \ t_i \leq t_{i+1}, i = 0,1,\cdots,N_a+1,
\end{aligned}
\end{equation}
with $f_a\coloneqq \sum_{i=0}^{N_a} \int_{t_{i}}^{t_{i+1}} \Tr\left[ \Sigma_{1,d}(t-t_{i}) \varphi_a(t)\right]dt$. From \cite[Proposition 8.5.12]{bernstein2009matrix}, we know that if $\Sigma_1 \geq \Sigma_2$, $\Tr[\Sigma_1 M] \geq \Tr[\Sigma_1M]$ for a positive semi-definite matrix $M$. Note that $\varphi_a(t)$ is positive semi-definite and $\Sigma_{1,d}(t) > \Sigma_{1,d}(t')$ for $t>t'$. Hence, $f_a^*(N_a)$ is a decreasing function of $N_a$, which aligns our intuition that the more observations received, the better the control would be. After solving ({\ref{eq:ObsrvationInstancesOp}}), it remains to find the optimal number of observations $N^*_a$ that minimizes $f_a(N_a) + O N_a$.  In the following theorem, we show that there always exists a minimizer for the optimization problem in (\ref{eq:ObsrvationInstancesOp}), which can be characterized by a set of necessary conditions.
\begin{theorem}\label{theorem:NecessaryCondition}
There always exists a solution, denoted by $N_a^*$ and $t^*_1,t^*_2,\cdots, t^*_{N^*_a}$, that solves $\min_{\mathcal{T}_a} \tilde{J}_a(\mathcal{T}_a)$. Furthermore, the optimal number of observations $N_a^*$ is bounded, i.e.,
\begin{equation}\label{eq:BoundedObservationNumber}
{N_a^* < \frac{1}{O} \int_{0}^{t_f} \Tr\left[ \Sigma_{1,d}(t) \varphi_a(t)\right]dt.}
\end{equation}
And the optimal observation instances $t_1^*,t_2^*,\cdots t^*_{N_a^*}$ need to satisfy
\begin{equation}\label{eq:FirstOrderNecessaryCondition}
\begin{aligned}
&\int_{t^*_{i-1}}^{t^*_{i}}\Tr\left[ e^{A_d(t^*_{i}-t)}C_dC_d'{e^{A_d(t^*_{i}-t)}}'\varphi_a(t^*_{i})\right]dt\\ 
=& \int_{t^*_{i}}^{t^*_{i+1}} \Tr\left[e^{A_d(t-t^*_{i})} C_dC_d' {e^{A_d(t-t^*_{i})}}'\varphi_a(t) \right]dt,\\
\end{aligned}
\end{equation}
for $i=1,2,\cdots, N_a^*$.
\end{theorem}
The proof is presented in Appendix \ref{proof:NecessaryCondition}. From Theorem \ref{theorem:NecessaryCondition}, we know that the optimal number of observations $N_a^*$ is bounded and inversely proportional to the observation cost $O$. From (\ref{eq:FirstOrderNecessaryCondition}), one can say that the optimal observation instances are spread out over the horizon $[0,t_f]$. Given a limited number of observations, it is unwise to allocate two observation instances in a short period of time. The effect of control is applied via $\varphi_a(t)$. For some period when $\varphi_a(t)$ is large, e.g., $\varphi_a(t) \geq \varphi(t')$, for $t\in [\tau_1,\tau_2]$ and $t'\notin [\tau_1,\tau_2]$. Then in this period, the attacker tend to observes more frequently. For example, if the goal of the attacker is to hit the asset at time $t_f$, the attacker may need to observe more frequently at the end of the game.

As is shown in Appendix \ref{proof:NecessaryCondition}, the differential of $f_a$ can be calculated analytically. The second-order differential can also be calculated analytically. Hence, we can leverage either first-order methods or second-order methods \cite{gill1991numerical} to numerically compute $t^*_{1}, t_{2}^*,\cdots, t^*_{N_a}$. Also, (\ref{eq:FirstOrderNecessaryCondition}) indicates that once $t_1^*$ is provided, $t_2^*,\cdots, t_{N_a}^*$ can be computed easily. Based on this feature, we propose a binary search algorithm that solves problem (\ref{eq:ObsrvationInstancesOp}). In Algorithm \ref{algorithm1}, we aim to find a $t_1^\star$ such that $\vert t_1^\star - t_{1}^*\vert <\epsilon/2$. Line $1$ initializes all the parameters in (\ref{eq:ObsrvationInstancesOp}). Line $2$ sets the initial low bound $t_{low}$ and upper bound $t_{up}$ of $t_1^*$ to be $0$ and $t_f$ respectively. The initial guess of $t_1$ is $(0+t_f)/2$. Line $5$ computes the left-hand side of (\ref{eq:FirstOrderNecessaryCondition}), which we rewrite as 
\begin{equation}\label{eq:FirstOrderNecessaryConditionLeft}
    l_a(t_{i-1},t_i)=\int_{t_{i-1}}^{t_{i}}\Tr\left[ e^{A_d(t_{i}-t)}C_dC_d'{e^{A_d(t_{i}-t)}}'\varphi_a(t_{i})\right]dt.
\end{equation}
Line $6$ computes the right-hand side integral in (\ref{eq:FirstOrderNecessaryCondition}) from $t_i$ to $t_f$, which we write as 
\begin{equation}\label{eq:FirstOrderNecessaryConditionRight}
    r_a(t_i,t_f)=\int_{t_{i}}^{t_{f}} \Tr\left[e^{A_d(t-t_{i})} C_dC_d' {e^{A_d(t-t_{i})}}'\varphi_a(t) \right]dt.
\end{equation}
Line $7$-$11$ says for any $t_i,i=1,2,\cdots,N_a$ that is computed based on our guess $t_1$, if $ r_a(t_i,t_f)<l_a(t_{i-1},t_i)$, then our guess $t_1$ is larger than $t_1^*$. Hence, we set the upper bound $t_{up}$ as $t_1$ and reset out guess $t_1$ as $t_1 = (t_{low} + t_1)/2$. Then we break the for loop and start with our new guess $t_1$. Line $12$ computes the next observation instance using (\ref{eq:FirstOrderNecessaryCondition}). Line $13$-$21$ says that when the for loop gets to $i=N_a$, we compute $t_{N_a+1}$. If $t_{N_a+1}<t_f$, our guess $t_1$ must be smaller than $t_1^*$. Hence, we set $t_{low} = t_1$, let our new guess to be $t_1 = (t_{up}+t_1)/2$, and breaks the for loop. If $t_{N_a+1} = t_f$ (it is impossible that $t_{N_a+1}>t_f$ due to our operations in Line $5$-$11$), then $t_1 = t_1^*$. Hence, we set $t_{low} = t_{up} =t_1$ to leave the while loop. Since the while ends when $\vert t_{up} - t_{low}\vert <\epsilon$,  we can ensure $\vert t_1^\star -t_1^* \vert<\epsilon/2$, where $t_1^*$ is the optimal first observation instance and $t_1^\star$ is the first observation instance found using Algorithm \ref{algorithm1}. The number of iterations needed for the while loop is less than $\min\{n\ |\ t_f/2^n\ \leq \epsilon\}$. For example, only $20$ iterations are needed to achieve $\epsilon = 10^{-5}$ when $t_f = 10$. Once $t_1^\star$ is obtained, the rest observation instances can be computed easily using  (\ref{eq:FirstOrderNecessaryCondition}). Note that with $f_a^*(N_a)$ being computed for some small $N_a$, a bound similar to yet tighter than (\ref{eq:BoundedObservationNumber}) can be developed. For example, when $f_a^*(N_a)$ is computed for $N_a=1,2,3$, if $N_a^*>3$, we have $f^*_a(3)+3O > O N_a^*$, i.e., $N_a^* -3 \leq f_a^*(3)/O$. Hence, we only need to compute $f_a^*(N_a)$ for a very limited number of $N_a$.

\begin{algorithm}[t]
\caption{Optimal Observation Instances Algorithm Based on Binary Search}\label{algorithm1}
\begin{algorithmic}[1]
\State Initialize $A_d$,$C_d$, $N_a$,$\varphi_a(\cdot)$, $t_f$, and tolerate$,\epsilon>0$
\State Set $t_{low}=0$, $t_0=0$, $t_{up}=t_f$, and $t_1 = (t_{up}+t_{low})/2$
\While {$|t_{up}-t_{low}|>\epsilon$}
\For{$i = 1,\cdots,N_a$}
\State  Compute $\textrm{val} = l_a(t_{i-1},t_i)$ defined in (\ref{eq:FirstOrderNecessaryConditionLeft})
\State Compute $\textrm{val}' = r_a(t_i,t_f)$ defined in (\ref{eq:FirstOrderNecessaryConditionRight})
\If{$\textrm{val}' < \textrm{val}$}
\State $t_{up} = t_{1}$
\State $t_1 = (t_{low}+t_1)/2$
\State \textbf{break}
\EndIf
\State Compute $t_{i+1}$ using (\ref{eq:FirstOrderNecessaryCondition})
\If{$i=N_a$} 
\If{$t_{i+1}<t_f$}
\State $t_{low} =t_1$
\State $t_{1} = (t_{up} + t_1)/2$
\State \textbf{break}
\Else
\State $t_{low} = t_{up} = t_1$
\EndIf
\EndIf
\EndFor
\EndWhile
\State \textbf{return} $t_1^\star = (t_{low} + t_{up})/2$
\end{algorithmic}
\end{algorithm}
In the next section, we provide case studies to demonstrate the computation and the effectiveness of observation instances and offer more insights.

\section{Case Studies}\label{sec:CaseStudy}
We consider a ``simple motion'' dynamics of two players  on a $2$-D plane. This dynamic model has been widely adopted by existing works \cite{garcia2020two, li2011defending, talebi2017cooperative,garcia2020defense}. The dynamics (\ref{eq:SeparateDynamics}) becomes 
\begin{equation}\label{eq:CaseStudyDynamics}
\begin{aligned}
d{x}_a &= b_a \cdot \Id_2 u_a(t)dt + C_a {dw_a(t)},\ \ \ x_a(0) = x_{a0},\\
d{x}_d &= b_d \cdot\Id_2 u_d(t)dt + C_d {dw_d(t)},\ \ \ x_d(0) = x_{d0},\\
\end{aligned}
\end{equation}
where $x_a\in\mathbb{R}^2$ and $x_d\in\mathbb{R}^2$, $b_a$ and $b_d$ are scalars that describe the maneuverability of the attacker and the defender respectively. Suppose the attacker and the defender only care the terminal state, i.e., $\omega_d^I = \omega_a^I =0$, and hence $Q=\mathbf{0}$.

\begin{proposition}
For system (\ref{eq:CaseStudyDynamics}) with $\omega_d^I = \omega_a^I = 0$, the Riccati equation admits a bounded closed-form solution
\begin{equation}\label{eq:CaseStufyRiccatiSolution}
K_{11}(t) = k(t)\begin{bmatrix}
\kappa_1(t) & -1\\
-1 & \kappa_2(t)
\end{bmatrix}\otimes \Id_2,
\end{equation}
where 
$$
\begin{aligned}
&\kappa_1(t) = -b_d^2(t-t_f)\omega_a +1 -\omega_a/\omega_b,\\
&\kappa_2(t) = b_a^2(t-t_f)\omega_a +1,\\
&k(t)=\\
& \frac{\omega_a\omega_d}{\left[ \omega_a b_a^2(t-t_f)+1\right]\left[ -\omega_a\omega_db_d^2(t-t_f)+\omega_d -\omega_a\right] -\omega_d}.\\
\end{aligned}
$$
The optimal observation instances $t_1^*,t_2^*,\cdots,t_{N_a}^*$ need to satisfy
\begin{equation}\label{eq:CaseStudyNecConditions}
k^2(t_i^*) (t_i^* - t^*_{i-1}) = \int_{t_i^*}^{t^*_{i+1}} k^2(t)dt.
\end{equation}
\end{proposition}
\begin{proof}
Under this setting, the riccati equation (\ref{eq:TruncatedSystemPara}) becomes
$$
\dot{K}_{11} =  - K_{11}\left(\hat{B}_d \hat{B}_d' - \hat{B}_a\hat{B}_a'\right)K_{11},\ \ \ \textrm{with }K_{11} = Q_{f,11}.
$$
Note that $Q_{f,11}$ is invertible. Indeed, we have
$$
Q^{-1}_{f,11} = \begin{bmatrix}
\frac{1}{\omega_a} & \frac{1}{\omega_a}\\
\frac{1}{\omega_a} & \frac{1}{\omega_a} - \frac{1}{\omega_d}\\
\end{bmatrix}\otimes \Id_2.
$$
Let $\Gamma = K_{11}^{-1}$ for some $[t_f- \Delta t, t_f]$. Since $K_{11} K_{11}^{-1} = \Id_4$, 
$$
\frac{d}{dt}K_{11}(t)K_{11}^{-1}(t) = \dot{K}_{11} K_{11}^{-1} +  K_{11} \dot{K}_{11}^{-1} =0,
$$
which gives $\dot{K}_{11}^{-1} = - K_{11}^{-1} \dot{K}_{11} K_{11}^{-1}$. Hence,
$$
\dot{\Gamma} = -K^{-1}_{11} \dot{K}_{11} K_{11}^{-1} = \hat{B}_a\hat{B}_a' - \hat{B}_d \hat{B}_d' =\begin{bmatrix}
b_a^2 & 0\\
0 & - b_d^2
\end{bmatrix} \otimes \Id_2.
$$ 
Since $\Gamma(t_f) = Q_{f,11}^{-1}$, we obtain
$$
\Gamma(t)  =\begin{bmatrix}
b_a^2 (t-t_f) + \frac{1}{\omega_a} & \frac{1}{\omega_a} \\
\frac{1}{\omega_a}  & -b_d^2(t-t_f) + \frac{1}{\omega_a}  - \frac{1}{\omega_d} 
\end{bmatrix}\otimes \Id_2.
$$
Since $K_{11}= \Gamma^{-1}$, we have (\ref{eq:CaseStufyRiccatiSolution}).

Then, we have $\varphi_a(t) = {K_{11}^{ur}}' \tilde{B}_a\tilde{B}_aK_{11}^{ur} = k^2(t)b_a^2\Id_2$ and similarly, $\varphi_b(t) = k^2(t) b_d^2 \Id_2$. Note that $A_a=A_d = 0$, which gives $\Sigma_{1,d}(t) = t C_d C_d'$ and $\Sigma_{2,a} = t C_a C_a'$. Hence,
$$
f_a^*(N_a) =  \Tr[ b_a^2 C_d C_d']\cdot \min_{t_1,\cdots, t_{N_a}} \sum_{i=0}^{N_a} \int_{t_i}^{t_{i+1}} k^2(t) ( t-t_i) dt.
$$
Furthermore, (\ref{eq:FirstOrderNecessaryCondition}) becomes
$$
\begin{aligned}
&\int_{t^*_{i-1}}^{t^*_{i}}\Tr\left[ C_dC_d'k^2(t_i^*) b_a^2 \right]dt= \int_{t^*_{i}}^{t^*_{i+1}} \Tr\left[C_dC_d'k^2(t) b_a^2 \right]dt,\\
&\Rightarrow k^2(t_i^*) (t_i^* - t^*_{i-1}) = \int_{t_i^*}^{t^*_{i+1}} k^2(t)dt.
\end{aligned}
$$
\end{proof}
Other parameters are set to be $\omega_d=3, C_d= C_a = 2\Id$, $b_a=b_d=0.5$. We use Algorithm \ref{algorithm1} to compute the optimal observation instances for the attacker. In fig. \ref{fig:RiccatiAndInstances}, we first plot the Riccati equation component $k^2(t)$ over $[0,t_f]$. The value of $k^2(t)$ captures the importance level of an observation at time $t$ for the attacker. For example, $k^2(t)$ attains its highest value around $5.3$. In an optimal solution, observations occur more frequently around time $5.3$. In fact, an observation at the beginning is not as valuable as an observation near the terminal time, as we can see from the curve of $k^2(\cdot)$. This is because $\omega^I_d =\omega^I_d =0$ and the attacker only cares about the relative positions between him/her and his/her opponents as well as the asset. We next present the optimal observation instances when the number of observations the attacker can take is limited to $N_a = 5$. The first observation occurs late at $t=2.8682$ and the fourth and the fifth observation instances are the closest. The distribution of the observation instances illustrates how $k^2(\cdot)$ affect the attacker's observation choices. The khaki box represents the left-hand side of (\ref{eq:CaseStudyNecConditions}) and the lavender area represents the right-hand side of (\ref{eq:CaseStudyNecConditions}). The necessary conditions (\ref{eq:CaseStudyNecConditions}) requires the areas of the two areas to be equal. This equality holds for each such neighboring areas.


\begin{figure}
    \centering
    \includegraphics[width=0.8\columnwidth]{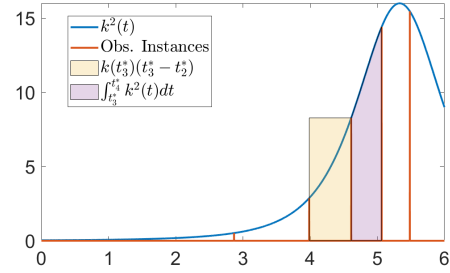}
    \caption{$k^2(t)$ is defined in (\ref{eq:CaseStufyRiccatiSolution}) and represented by blue line. The optimal observation instances $t_i^*,i=1,\cdots,5$ are represented by red impulses with height adjusted to $k^2(t_i^*)$. The area of the khaki box being equal to that of the lavender area illustrates the necessary conditions (\ref{eq:CaseStudyNecConditions}).   }
    \label{fig:RiccatiAndInstances}
\end{figure}

\begin{figure}
    \centering
    \includegraphics[width=1\columnwidth]{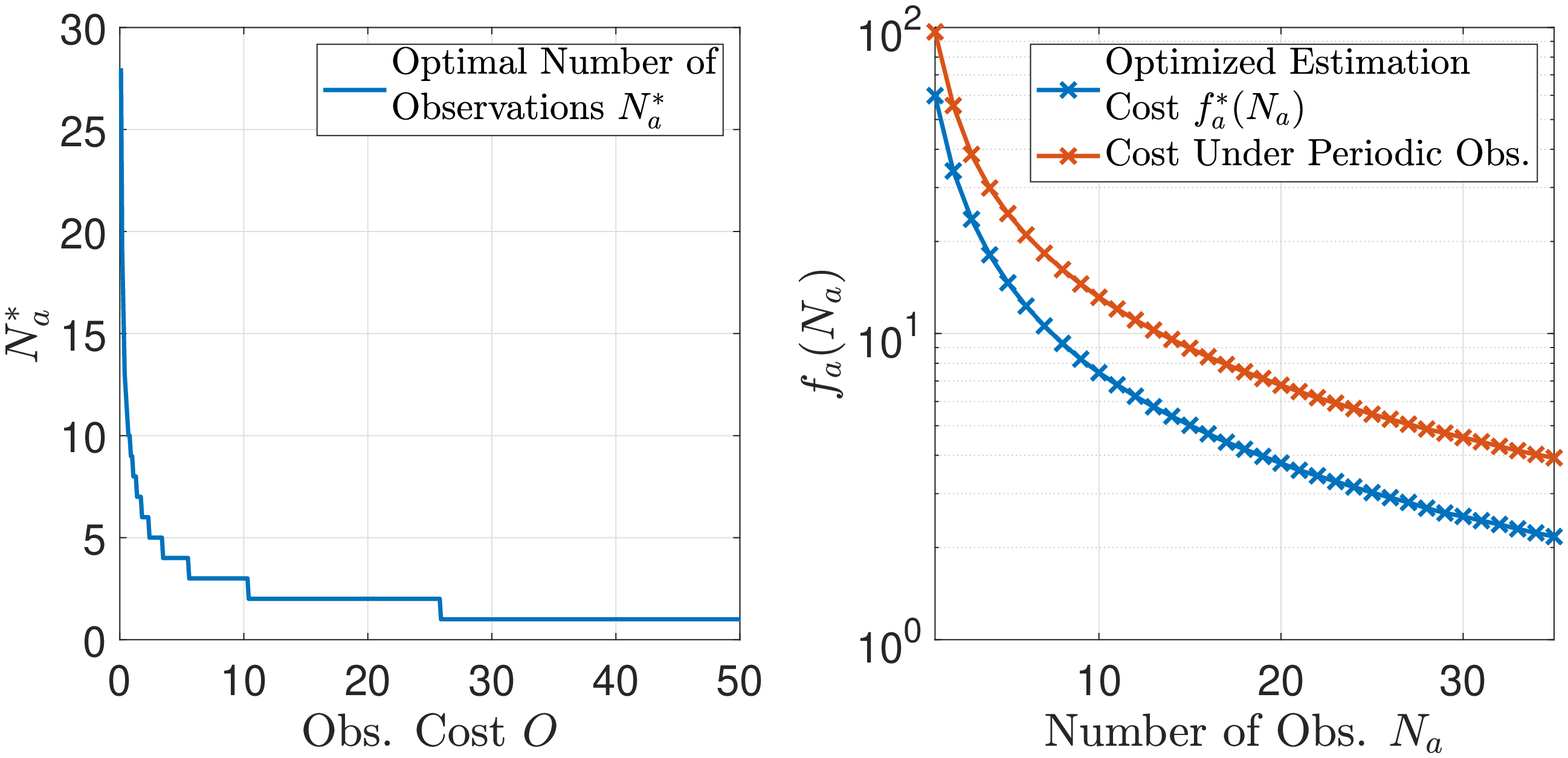}
    \caption{The optimal number of observations for the attacker under different observation costs (left). Given the number of observations $N_a$ allowed, the costs $f_a$ under two choices of observation instances: the periodic one (marked red) and the optimal one (marked blue) computed using Algorithm \ref{algorithm1} (right).}
    \label{fig:OptimalNumObservations}
\end{figure}

In Fig. \ref{fig:OptimalNumObservations}, the left figure shows that the optimal number of observations is inversely proportional to the observation cost. That means the defender can limit the performance of the attacker by increasing the attacker's observation cost. The increase of observation cost can be done, for example, by leveling up the defender's stealthiness. We compare two sets of observation instances in the right plot of Fig. \ref{fig:OptimalNumObservations}. One is the optimal observation instances $t_1^\star,t_2^\star,\cdots,t_{N_a}^{\star}$ calculated using Algorithm \ref{algorithm1} and one is the periodic observation instances chosen as $t_i= i\cdot t_f/(N_a+1)$. The optimal observation instances induce less cost than the periodic observation instances given the same number of observations $N_a$. More precisely, by adopting the optimal observation instances, the attacker can reduce at least $30\%$ of the cost generated by adopting the periodic observation instances. Similar results can also be obtained for the defender.

\begin{figure}
    \centering
    \includegraphics[width=0.9\columnwidth]{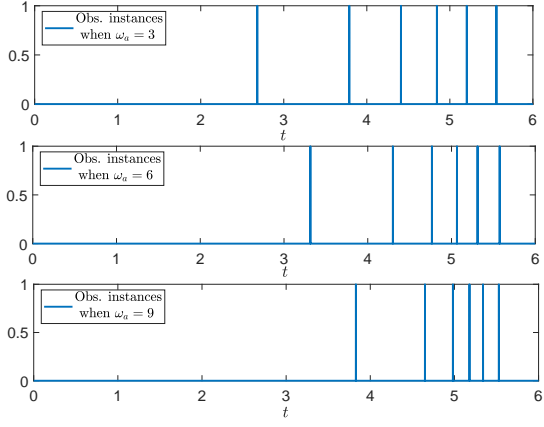}
    \caption{When $N_a=5$, the optimal observation instances under different $\omega_a$.}
    \label{fig:ObservationInstances}
\end{figure}

In Fig. \ref{fig:ObservationInstances}, we present the optimal time instances under different $\omega_a$. Note that $\omega_a$ is the weight assigned to the terminal distance between the attacker and the defender, as one can see in (\ref{Eq:CostFunctional}). Even though the location of the asset is fully known to the attacker and the observations are made to observe the defender's location, the weight assigned to the distance between the attacker and the asset still affect the choice of the optimal observation instances. As $\omega_a$ increases, the attacker tends to observe late and observation instances get closer to each other. This is due to the fact that $\omega_a$ increases the terminal cost while the transient cost is zero because $\omega_a^I=\omega_d^I =0$. 

\section{Conclusions}

In this paper, we look into a DADG with partial information and selected observations. Due to the LQG formulation, both players' observation choices are independent of the asset's trajectory. The Nash observation game can be decoupled into two separate optimization problems. A set of necessary conditions has been exploited to develop an effective algorithm to compute the optimal observation instances. Case studies show that the optimal observation instances outperform the periodic observation strategy by at least $30\%$. Future works can focus on studying noisy observation settings and investigating statistic results regarding success rate and capturability rate.










\bibliography{references}

\begin{thebibliography}{10}
\providecommand{\url}[1]{#1}
\csname url@samestyle\endcsname
\providecommand{\newblock}{\relax}
\providecommand{\bibinfo}[2]{#2}
\providecommand{\BIBentrySTDinterwordspacing}{\spaceskip=0pt\relax}
\providecommand{\BIBentryALTinterwordstretchfactor}{4}
\providecommand{\BIBentryALTinterwordspacing}{\spaceskip=\fontdimen2\font plus
\BIBentryALTinterwordstretchfactor\fontdimen3\font minus
  \fontdimen4\font\relax}
\providecommand{\BIBforeignlanguage}[2]{{%
\expandafter\ifx\csname l@#1\endcsname\relax
\typeout{** WARNING: IEEEtran.bst: No hyphenation pattern has been}%
\typeout{** loaded for the language `#1'. Using the pattern for}%
\typeout{** the default language instead.}%
\else
\language=\csname l@#1\endcsname
\fi
#2}}
\providecommand{\BIBdecl}{\relax}
\BIBdecl

\bibitem{liu2007pop}
Y.~Liu, J.~B. Cruz, and C.~J. Schumacher, ``Pop-up threat models for persistent
  area denial,'' \emph{IEEE Transactions on Aerospace and Electronic Systems},
  vol.~43, no.~2, pp. 509--521, 2007.

\bibitem{weintraub2020introduction}
I.~E. Weintraub, M.~Pachter, and E.~Garcia, ``An introduction to
  pursuit-evasion differential games,'' in \emph{2020 American Control
  Conference (ACC)}.\hskip 1em plus 0.5em minus 0.4em\relax IEEE, 2020, pp.
  1049--1066.

\bibitem{talebi2017cooperative}
S.~Talebi, M.~A. Simaan, and Z.~Qu, ``Cooperative, non-cooperative and greedy
  pursuers strategies in multi-player pursuit-evasion games,'' in \emph{2017
  IEEE Conference on Control Technology and Applications (CCTA)}.\hskip 1em
  plus 0.5em minus 0.4em\relax IEEE, 2017, pp. 2049--2056.

\bibitem{li2011defending}
D.~Li and J.~B. Cruz, ``Defending an asset: a linear quadratic game approach,''
  \emph{IEEE Transactions on Aerospace and Electronic Systems}, vol.~47, no.~2,
  pp. 1026--1044, 2011.

\bibitem{garcia2020two}
E.~Garcia, D.~W. Casbeer, M.~Pachter, J.~W. Curtis, and E.~Doucette, ``A
  two-team linear quadratic differential game of defending a target,'' in
  \emph{2020 American Control Conference (ACC)}.\hskip 1em plus 0.5em minus
  0.4em\relax IEEE, 2020, pp. 1665--1670.

\bibitem{garcia2018design}
E.~Garcia, D.~W. Casbeer, and M.~Pachter, ``Design and analysis of
  state-feedback optimal strategies for the differential game of active
  defense,'' \emph{IEEE Transactions on Automatic Control}, vol.~64, no.~2, pp.
  553--568, 2018.

\bibitem{garcia2020defense}
------, ``Defense of a target against intelligent adversaries: A linear
  quadratic formulation,'' in \emph{2020 IEEE Conference on Control Technology
  and Applications (CCTA)}.\hskip 1em plus 0.5em minus 0.4em\relax IEEE, 2020,
  pp. 619--624.

\bibitem{liang2019differential}
L.~Liang, F.~Deng, Z.~Peng, X.~Li, and W.~Zha, ``A differential game for
  cooperative target defense,'' \emph{Automatica}, vol. 102, pp. 58--71, 2019.

\bibitem{pachter2017differential}
M.~Pachter, E.~Garcia, and D.~W. Casbeer, ``Differential game of guarding a
  target,'' \emph{Journal of Guidance, Control, and Dynamics}, vol.~40, no.~11,
  pp. 2991--2998, 2017.

\bibitem{cooper1971optimal}
C.~Cooper and N.~Hahi, ``An optimal stochastic control problem with observation
  cost,'' \emph{IEEE Transactions on Automatic Control}, vol.~16, no.~2, pp.
  185--189, 1971.

\bibitem{huang2019continuous}
Y.~Huang, V.~Kavitha, and Q.~Zhu, ``Continuous-time markov decision processes
  with controlled observations,'' in \emph{2019 57th Annual Allerton Conference
  on Communication, Control, and Computing (Allerton)}.\hskip 1em plus 0.5em
  minus 0.4em\relax IEEE, 2019, pp. 32--39.

\bibitem{olsder1977observation}
G.~J. Olsder, ``On observation costs and information structures in stochastic
  differential games,'' in \emph{Differential Games and Applications}.\hskip
  1em plus 0.5em minus 0.4em\relax Springer, 1977, pp. 172--185.

\bibitem{huang2020infinite}
Y.~Huang and Q.~Zhu, ``Infinite-horizon linear-quadratic-gaussian control with
  costly measurements,'' \emph{arXiv preprint arXiv:2012.14925}, 2020.

\bibitem{maity2016strategies}
D.~Maity and J.~S. Baras, ``Strategies for two-player differential games with
  costly information,'' in \emph{2016 13th International Workshop on Discrete
  Event Systems (WODES)}.\hskip 1em plus 0.5em minus 0.4em\relax IEEE, 2016,
  pp. 211--216.

\bibitem{maity2017linear}
------, ``Linear quadratic stochastic differential games under asymmetric value
  of information,'' \emph{IFAC-PapersOnLine}, vol.~50, no.~1, pp. 8957--8962,
  2017.

\bibitem{huang2020cross}
Y.~Huang, Z.~Xiong, and Q.~Zhu, ``Ccross-layer coordinated attacks on
  cyber-physical systems: A lqg game framework with controlled observations,''
  in \emph{2021 European Control Conference (ECC)}.\hskip 1em plus 0.5em minus
  0.4em\relax IEEE, 2021, available at: \url{https://arxiv.org/abs/2012.02384}.

\bibitem{chen2019covariance}
Y.~Chen, T.~T. Georgiou, and M.~Pavon, ``Covariance steering in zero-sum
  linear-quadratic two-player differential games,'' in \emph{2019 IEEE 58th
  Conference on Decision and Control (CDC)}.\hskip 1em plus 0.5em minus
  0.4em\relax IEEE, 2019, pp. 8204--8209.

\bibitem{huang2021pursuit}
Y.~Huang and Q.~Zhu, ``A pursuit-evasion differential game with strategic
  information acquisition,'' \emph{arXiv preprint arXiv:2102.05469}, 2021.

\bibitem{engwerda2005lq}
J.~Engwerda, \emph{LQ dynamic optimization and differential games}.\hskip 1em
  plus 0.5em minus 0.4em\relax John Wiley \& Sons, 2005.

\bibitem{bernstein2009matrix}
D.~S. Bernstein, \emph{Matrix mathematics: theory, facts, and formulas}.\hskip
  1em plus 0.5em minus 0.4em\relax Princeton university press, 2009.

\bibitem{gill1991numerical}
P.~E. Gill, W.~Murray, M.~H. Wright \emph{et~al.}, \emph{Numerical linear
  algebra and optimization}.\hskip 1em plus 0.5em minus 0.4em\relax
  Addison-Wesley Redwood City, CA, 1991, vol.~1.

\end{thebibliography}

\bibliographystyle{IEEEtran}

\appendix
\subsection{Proof of Theorem \ref{theorem:NecessaryCondition}}\label{proof:NecessaryCondition}
\begin{proof}
First, note that $f_a(t_1,t_2,\cdots,t_{N_a})$ is differentiable over $t_1,t_2,\cdots,t_{N_a}$. Indeed, using Leibniz integral rule, we obtain
{\small$$
\begin{aligned}
&\frac{\partial}{\partial t_i} f_a(t_1,\cdots,t_{N_a})\\ 
=& \frac{\partial}{\partial t_i}\Bigg[ \int_{t_{i-1}}^{t_i} \Tr[\Sigma_{1,d}(t-t_{i-1})\varphi_a(t)]dt &\\
&-\int_{t_{i}}^{t_i+1} \Tr[\Sigma_{1,d}(t-t_{i})\varphi_a(t)]dt \Bigg]\\
=&\int_{t_{i-1}}^{t_{i}}\Tr\left[ e^{A_d(t_{i}-t)}C_dC_d'{e^{A_d(t_{i}-t)}}'\varphi_a(t_{i})\right]dt\\
&-\int_{t_{i}}^{t_{i+1}} \Tr\left[e^{A_d(t-t_{i})} C_dC_d' {e^{A_d(t-t_{i})}}'\varphi_a(t) \right]dt.
\end{aligned}
$$}
Hence, $f_a$ is continuous over $t_{1},t_2,\cdots, t_{N_a}$. Besides, the constraint set in (\ref{eq:ObsrvationInstancesOp}) is closed and bounded (hence compact). Then by Weierstrass extreme value theorem, there exists at least one minimizer $t^*_1,t^*_2,\cdots, t^*_{N_a}$ for problem (\ref{eq:ObsrvationInstancesOp}) for any given $N_a$. To obtain the first-order optimality condition, we let $\frac{\partial }{\partial t_i} f_a(t_1,t_2,\cdots,t_{N_a})=0$, which yields (\ref{eq:FirstOrderNecessaryCondition}).
Further, notice that 
$$
\tilde{J}_a(\mathcal{T}^*_a) \leq \tilde{J}_a(\varnothing) = f_a^*(0) = \int_{0}^{t_f} \Tr[\Sigma_{1,d}(t)\varphi_a(t)],
$$
and
$
\tilde{J}_a(\mathcal{T}^*_a) > O N_a^*.
$
Hence, we have $O N_a^* < \int_{0}^{t_f} \Tr[\Sigma_{1,d}(t)\varphi_a(t)]$, which gives (\ref{eq:BoundedObservationNumber}). Since $N_a^* \in\mathbb{N}$ is bounded and $f^*(N_a)$ for any $N_a$, there always exist $N_a^*$ and $t_1^*,\cdots, t_{N_a}^*$ that solve $\min_{\mathcal{T}_a} \tilde{J}_a(\mathcal{T}_a)$.
\end{proof}

\end{document}